\documentclass[a4paper,12pt]{amsart}

\usepackage{amsmath, amsthm, amsfonts, amssymb} %basic mathematical fonts and symbols
\usepackage{graphicx} %inclusion of pictures
\usepackage{fnpct} %proper footnotes spacing after commas and dots
\usepackage{hyperref}						%references are clickable
\hypersetup{colorlinks=true, linkcolor=red, urlcolor=blue, citecolor=blue, pdfpagemode=UseNone, pdfstartview=FitH, bookmarksopen=true}
\hypersetup{pdftitle=Title}
\usepackage[utf8]{inputenc}

\usepackage{fdsymbol}

\newcommand{\R}{\mathbb{R}} 	%real numbers
\newcommand{\Z}{\mathbb{Z}} 	%integers
\newcommand{\Q}{\mathbb{Q}} 	%rational numbers
\newcommand{\F}{\mathcal{F}} 	%family of sets

\DeclareMathOperator\sd{sd}

\newcommand{\pth}[1]{\left( #1 \right)}
\newcommand{\skel}[2]{#2^{(#1)}}

\DeclareMathOperator{\conv}{conv}

\newtheorem{theorem}{Theorem}
\newtheorem{lemma}[theorem]{Lemma}
\newtheorem{proposition}[theorem]{Proposition}
\theoremstyle{definition}
\newtheorem{definition}[theorem]{Definition}

\begin{document}

%opening
\title{Helly-type problems from a topological perspective}
\author[Pavel Paták]{Pavel Paták}
\author[Zuzana Patáková]{Zuzana Patáková}
\address[Pavel Paták]{Faculty of Information Technology, Czech Technical University in Prague, Thákurova 9, 160~00~Praha, Czech Republic}
\email{patakpav@cvut.cz}
\address{Department of Applied Mathematics, Faculty of Mathematics and Physics, Charles University, Malostranské náměstí 25, 118 00 Praha, Czech Republic}
\email{patak@kam.mff.cuni.cz}
\address[Zuzana Patáková]{Department of Algebra, Faculty of Mathematics and Physics, Charles University, Sokolovská~83, 186 75 Praha, Czech Republic}
\email{patakova@karlin.mff.cuni.cz}
\thanks{Both authors were supported by GA\v{C}R grant 25-16847S}

\begin{abstract}
 We discuss recent progress on topological Helly-type theorems and their variants. We provide an overview of two different proof techniques, one based on the nerve lemma, while the other on non-embeddability.
\end{abstract}

\maketitle

\section{Introduction}

Helly's theorem, about 100 years old now, is one of the cornerstones of convex geometry and still attracts a lot of attention, generates interesting open problems \cite{Barany_Kalai} and has been applied in different areas; see \cite{DeLorea_Goaoc_Meunier_Mustafa, amenta2016hellystheoremnewvariations} for recent surveys.

\begin{theorem}[Helly's theorem~\cite{h-umkkm-23}]\label{thm:helly} Given a finite family $\mathcal F$ of convex sets in $\R^d$, either $\bigcap \mathcal F\neq \emptyset$, or one can find a subfamily $\mathcal G\subseteq \mathcal F$ that has at most $d+1$ elements and satisfies $\bigcap \mathcal G=\emptyset$.
\end{theorem}

One readily observes that this property extends beyond convex sets, as it is preserved by any homeomorphism of the underlying space.
In this survey we study to which extent convexity in Helly's theorem can be replaced by topological conditions. 
More precisely, we focus on two different techniques for proving topological Helly-type results. Namely, Section \ref{sec:nerve}~provides an overview of Helly-type theorems that were proved using the nerve lemma or its variants. Section~\ref{sec:nerveprop} summarizes how similar methods have been used to prove colorful and fractional modifications of Helly's theorem. 
Section~\ref{sec:embed} is devoted to Helly-type results from the perspective of non-embeddability arguments. Since there has been significant progress in this direction in the past few years, this section goes into more detail and describes how various ideas fit together to enable these new and more general results. In Section \ref{sec:open} we list some related open problems.

\medskip

In order to present the topic concisely, we start with a few definitions that allow for a somewhat nicer formulation of Helly's theorem and the related results. 
\begin{definition}[Helly number \cite{Danzer_Grunbaum_Klee_Helly_number}]
Let $\mathcal C$ be a family of sets.
Then we define its \emph{Helly number} $h(\mathcal C)$ as the smallest natural number $h$ with the following property.
For each finite subfamily $\mathcal F$ of $\mathcal C$ with $\bigcap \mathcal F=\emptyset$ there exists $\mathcal G\subseteq\mathcal F$ with $\bigcap \mathcal G=\emptyset$ and $|\mathcal G|\leq h$.
If there is no such number, we set $h(\mathcal C):=+\infty$.
\end{definition}
With this definition, Helly's theorem states that the Helly number of convex sets in $\R^d$ is at most $d+1$. Moreover, this bound is tight as can be easily seen by considering hyperplanes in general position.

Given a positive integer $h$,
we want to find properties of $\mathcal C$ that ensure $h(\mathcal C)\leq h$.
As mentioned, we study this question in Sections~\ref{sec:nerve} and \ref{sec:embed}.
In Section~\ref{sec:nerve} we look at the results that are based on the so-called nerve lemma and its variants. The advantage of this approach is that the obtained bounds are usually tight. The disadvantage comes from the fact that the sets need to satisfy the assumptions of the nerve lemma, typically they need to be all open or all closed and their homology must be trivial in certain dimensions.

In Section~\ref{sec:embed} we look at topological Helly-type theorems that are derived from non-embeddability results.
The advantage of this approach is that the extra assumptions of the nerve lemma can be lifted. Thus, one can bound the Helly numbers even for sets that are neither open nor closed and have non-trivial homology groups. 
The drawback is that if we allow nontrivial homology, the obtained bounds on the Helly number become enormous very quickly and are probably far from being optimal.

\bigskip

Let us now present another definition and the related reformulation of Helly's theorem.
\begin{definition}[Nerve of a family of sets \cite{nervedef}]
Given a collection $\mathcal C$,
its \emph{nerve $\mathcal N(\mathcal C)$}
is defined as the collection of all finite subfamilies $\mathcal H$ of $\mathcal C$
with $\bigcap\mathcal H\neq\emptyset$.
Formally,
\[
\mathcal N(\mathcal C):=\left \{\mathcal H\subseteq \mathcal C\mid \bigcap \mathcal H\neq \emptyset, \mathcal H\text{ finite}\right\}.
\]
In this definition, we consider $\bigcap \emptyset$ to be always non-empty, i.e. $\emptyset$ always belongs to $\mathcal N(\mathcal C)$.
\end{definition}
Clearly, $\mathcal N(\mathcal C)$ is an abstract simplicial complex, that is, it is a collection of finite sets and if $\mathcal X\subseteq\mathcal Y$ and $\mathcal Y\in\mathcal N(\mathcal C)$, then also $\mathcal X\in \mathcal C$.

\begin{definition}[Missing face]
If $K$ is an abstract simplicial complex
and $H\notin K$ is a finite set with $|H|=k+1$ such that every proper subset $S$ of $H$ belongs to $K$, we say that $S$ is a \emph{missing $k$-face} of $K$.
\end{definition}
Using this definition, Helly's theorem can be restated in the following form.
If $\mathcal C$ denotes the family of all convex sets in $\R^d$, then $\mathcal N(\mathcal C)$ does not contain any missing $k$-face with $k>d$.

In Section~\ref{sec:nerveprop} we study the interaction between the topological properties of $\mathcal N(\mathcal C)$, the Helly number of $\mathcal C$ and other related parameters, such as the fractional Helly number or the colorful Helly number.

\subsection*{Notation and preliminaries.}
We write $\bigcap \mathcal F$ as a shorthand for $\bigcap_{F \in \mathcal F} F$.  By convention, we consider $\bigcap\emptyset$ to be always non-empty.
Similarly, $\bigcup \mathcal G$ means $\bigcup_{G \in \mathcal G} G$ and $\bigcup \emptyset = \emptyset$. For a set $S$ in an Euclidean space, $\operatorname{conv} S$ denotes its convex hull.

Throughout the text we use standard terminology from algebraic topology: The symbol $\Delta_n$ means the standard $n$-dimensional simplex with vertices $e_1,e_2,\ldots, e_{n+1}$.
If $K$ is a simplicial complex, $K^{(k)}$ denotes its $k$-skeleton, that is, the union of all its faces of dimension at most $k$. 
The symbol $\Z_2$ stands for $\Z/2\Z$, that is the ring of integers modulo two.
If $S$ is a subset of a topological space $X$,  $\partial S$ denotes its boundary.
The $i$th homology group of a topological space $X$ with coefficients in a ring $R$ is denoted $H_i(X;R)$ and its reduced version $\widetilde{H}_i(X;R)$. The most usual choices are $R=\Z$, $R=\Z_2$, or $R=\Q$.
If $R$ is a field, the dimension of $\widetilde{H}_i(X;R)$ as a vector space over $R$ is called the $i$th reduced Betti number and it is denoted $\widetilde{\beta}_i(X;R)$.
It follows from the definition of homotopy groups that the $k$-th homotopy group $\pi_k(X)$ of $X$ is trivial if every continuous map $f\colon \partial \Delta_{k+1}\to X$ can be extended to a continuous map $f\colon \Delta_{k+1}\to X$.
A topological space $X$ is called $d$-connected if $\pi_k(X)$ is trivial for all $k\leq d$.
If $X$ is $d$-connected, then $\widetilde{H}_i(X; R)=0$ for all $i\leq d$ and all rings $R$.

For a further introduction to algebraic topology, we refer to the classical texts~\cite{Hatcher:AlgebraicTopology-2002,Munkres:AlgebraicTopology-1984}.

\section{Topological Helly-type theorems from the nerve lemma}\label{sec:nerve}
The first Helly-type theorem from topological assumptions was given by Helly himself.
\begin{theorem}\cite{h-usvam-30}\label{thm:hellyTop}
Let $\mathcal C$ be a family of open sets in $\R^d$ such that the intersection of any subfamily is either contractible\footnote{A set is contractible if it can be continuously shrunken to a single point. } or empty, then $h(\mathcal C)\leq d+1$. 
\end{theorem}

This can be easily seen as a direct consequence of the nerve theorem.
\begin{theorem}\cite{nerveLeray,nerveBorsuk,nerveWeil,nerveMcCord}\label{thm:nerveLeray}
Let $\mathcal C$ be a family of open sets in a Haussdorff topological space
such that the intersection of any non-empty finite subfamily $\mathcal D\subseteq \mathcal C$ is either empty or contractible,
then $\mathcal N(\mathcal C)$ is homotopy equivalent to $\bigcup \mathcal C$.
\end{theorem}
We note that the original nerve theorems by Leray~\cite{nerveLeray} and Borsuk~\cite{nerveBorsuk} considered closed sets. 
The open sets were considered later in the works of Weil~\cite{nerveWeil} or McCord~\cite{nerveMcCord}. For a more detailed historical overview of the nerve theorem, we refer to~\cite{nerveSurvey}.

\begin{proof}[Proof of Theorem~\ref{thm:hellyTop} from Theorem~\ref{thm:nerveLeray}]
If there is a missing $k$-face $S$ in $\mathcal N(\mathcal C)$, the nerve theorem implies that the union of the sets from $S$ has a non-trivial $(k-1)$th homology as homotopically equivalent spaces have isomorphic homology groups. 
However, for any open set $G$ in $\R^d$, $H_i(G;\Z)=0$ for $i\geq d$.
This shows $k-1\leq d-1$, as desired.
\end{proof}
There have been several generalizations of the Helly-type results that were obtained by considering stronger forms of the nerve theorem. 
Probably the most general formulation in this direction is due to Montejano. 
\begin{theorem}[\cite{Montejano-Berge13}]\label{thm:montejano}
Let $X$ be a (paracompact Hausdorff) topological space with the property that $\widetilde{H}_t(U;\Z)=0$ for $t \geq d$ and for every open subset $U$ of $X$. Then a finite family $\F$ of open sets in $X$ has a non-empty intersection if for any $1\leq j\leq d+1$ and any $F_1, \ldots, F_j \in \F$, \[\widetilde{H}_{d-j}\left(F_1 \cap F_2 \cap \cdots \cap F_j;\Z\right)=0.\]
(Here we note that by definition $\widetilde{H}_{-1}(G;\Z)=0$ means that $G$ is non-empty.)
\end{theorem}
We note that the result itself was known before Montejano published it. For example, it follows from Theorem~\ref{thm:multinerveTwo} proved in \cite{multinerve13}, whose first version predates Montejano's result by several years. What is novel is Montejano's approach, where he does not use nerve theorem directly. Instead, he computes the homology using the long exact sequence of Mayer-Vietoris and induction, which provides a much simpler proof.

So far, the theorems we presented only dealt with families in which the relevant homology groups were trivial. In the next theorem $\widetilde H_0$ is allowed to be non-zero. 
Following the authors of~\cite{multinerve13}, we first present a version of the result that is easier to read but does not attain full generality.
\begin{theorem}[\cite{multinerve13}]\label{thm:multinerveOne}
Let $X$ be a locally arc-wise connected\footnote{A set $U$ is \emph{arc-wise connected} if for every two points $u,v\in U$ there is a continuous map $\gamma\colon [0,1]\to U$ with $\gamma(0)=u$ and $\gamma(1)=v$ such that $\gamma$ is a homeomorphism of $[0,1]$ onto its image.
A topological space $X$ is \emph{locally arc-wise connected}, if for every point $p\in X$ and every open set $U\ni x$, there is an open set $V$ that is arc-wise connected and satisfies $p\in V\subseteq U$.} topological space and $d$ be the smallest integer such that for all $j\geq d$ every open set $G\subseteq X$ satisfies $\widetilde H_j(G;\Q)=0$.

Suppose that $\mathcal F$ is a family of open subsets of $X$, such that for each non-empty $\mathcal G\subseteq \mathcal F$,
$\bigcap \mathcal G$ has at most $r$ connected components and satisfies $\widetilde{H}_i(\bigcap \mathcal G;\Q)=0$ for all $i\geq 1$.
Then $h(\mathcal F)\leq r(d+1)$.
\end{theorem}
To state the theorem in full generality, we introduce the following notion.
\begin{definition}[\cite{multinerve13}]
Let $\mathcal F$ be a family of open sets in a locally arc-wise connected topological space.
We say that $\mathcal F$ is \emph{acyclic with slack $s$} if for every non-empty subfamily $\mathcal G\subseteq\mathcal F$,
\[
\widetilde{H}_i(\bigcap \mathcal G;\Q)=0\qquad\text{for all }i\geq \max(1, s-|\mathcal G|).
\]
\end{definition}
\begin{theorem}[\cite{multinerve13}]\label{thm:multinerveTwo}
Let $X$ be a locally arc-wise connected topological space and $d$ be the smallest integer such that for all $j\geq d$ and all open sets $G\subseteq X$, $\widetilde H_j(G;\Q)=0$.

Assume that $\mathcal F$ is an acyclic family of open sets in $X$ with slack $s$ such that for every subfamily $\mathcal G\subseteq \F$
of at least $t$ sets, $\bigcap \mathcal G$ has at most $r$ connected components,
then 
\[h(\mathcal F)\leq r\cdot \left(\max(d,s,t)+1\right).\]
\end{theorem}
In order to prove these theorems, the authors introduce the multinerve $\mathcal M(\mathcal F)$. The multinerve resembles the nerve, but it has the extra crucial property that iff $\bigcap \mathcal G$ has $k$ connected components, the intersection in the multinerve is not represented by one simplex, but by $k$ distinct simplices.
It follows that in a multinerve distinct simplices may have the same set of vertices; see Figure~\ref{fig:multinerve} for illustration.
\begin{figure}
\begin{center}
\includegraphics{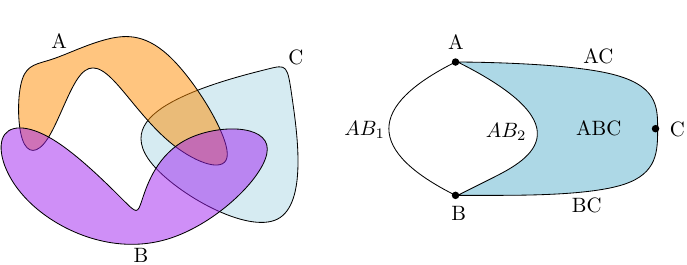}

A collection of three sets (on the left) and the corresponding multinerve (on the right).
Because the intersection of $A$ and $B$ consists of two components, there are two edges $AB_1$ and $AB_2$ that connect the points $A$ and $B$. The intersection $A\cap B\cap C$ only contains one of the connected components of $A\cap B$. Thus in the multinerve, there is only one triangle $ABC$ bounded by edges $AB_2$, $AC$ and $BC$.
\caption{An example of a multinerve}\label{fig:multinerve}
\end{center}
\end{figure} 
The multinerve is thus no longer a simplicial complex; instead, it is a simplicial poset. For technical details and proper definitions, we refer to~\cite{multinerve13}.

There is a variant of the nerve lemma for this construction.
\begin{theorem}[Multinerve theorem~\cite{multinerve13}]
Let $X$ be a locally arc-wise connected topological space and let $\mathcal F$ be a family of open subsets of $X$.
If $\mathcal F$ is acyclic with slack $s$,
then
$\widetilde{H}_j(\mathcal M(\mathcal F))\cong \widetilde{H}_j(\bigcup \mathcal F)$ for $j=0$ and for all $j\geq s$.
\end{theorem}
The proof of Theorem~\ref{thm:multinerveTwo} uses a spectral sequence that relates the homology of $\mathcal N(\mathcal F)$ and the homology of $\mathcal M(\mathcal F)$.

\bigskip
As usual in life, nothing is all-mighty, and currently known nerve lemmata are not an exception. We conclude this section by discussing the limitations of the presented approach.

 To put the multinerve theorem in broader context, one can see the multinerve construction and the multinerve theorem as a special case of homotopy colimit. Given a family $\mathcal F$ of sets in a topological space $X$, the homotopy colimit is a construction that builds a nice topological space $T$ that resembles $\bigcup \F$ based solely on the homotopy types of intersections $\bigcap \mathcal G$, where $\mathcal G$ ranges over all non-empty subfamilies of $\mathcal F$. For details we refer to~\cite{nerveSurvey}.
 The homotopy colimit $T$ is then related to the nerve $\mathcal N(\mathcal F)$
 through a spectral sequence. Unfortunately, this spectral sequence is very hard to analyze in full generality. So far, it has only been done if the relevant homology groups are trivial, which led to Theorem~\ref{thm:multinerveTwo}. 

\section{Topological properties of the nerve}\label{sec:nerveprop}
In this part of the survey, we look at various properties related to the Helly number
and to various topological conditions of the nerve that imply these properties.

\begin{definition}[$d$-collapsibility \cite{wegner75}]
Let $K$ be a simplicial complex.
Let $\sigma$ be a face of $K$ of dimension $<d$ that is contained in a unique facet\footnote{Facets are inclusion-wise maximal faces of $K$.} of $K$.
Let $L$ be the complex obtained from $K$ by removing $\sigma$ and all faces that contain it.
Then we say that $L$ is obtained from $K$ by an \emph{elementary $d$-collapse}.

A complex $K$ is \emph{$d$-collapsible}, if there is a sequence of elementary $d$-collapses that reduces it to a single point. 
\end{definition}
If $K$ is $d$-collapsible, then any induced subcomplex of $K$ is also $d$-collapsible.
This is because the collapse sequence for $K$ also induces the collapse sequence of the induced subcomplex.
Thus, if $\mathcal N(\mathcal F)$ is $d$-collapsible, then there is no missing $k$-face in $\mathcal N(\mathcal F)$ with $k>d$.
Consequently, any family $\mathcal F$ with $d$-collapsible nerve has Helly number at most $d+1$.  
By combining this fact with the result of Wegner~\cite{wegner75} that the nerve of any finite family of convex sets in $\R^d$ is $d$-collapsible,
we obtain yet another proof of Helly's theorem (Theorem~\ref{thm:helly}).
Unfortunately, checking $d$-collapsibility is an NP-complete problem, already for $d\geq 4$, see~\cite{dcollapse}.

Closely related to $d$-collapsibility is the notion of $d$-Lerayness.
\begin{definition}
Let $K$ be a simplicial complex
and $R$ be a ring of coefficients. 
We say that $K$ is \emph{$d$-Leray} over $R$, if for each induced subcomplex $L$ and each $i\geq d$, $\widetilde{H}_i(L;R)=0$.
\end{definition}
Since elementary $d$-collapses do not change the homology groups $\widetilde{H}_i(L;R)$ with $i\geq d$, we immediately see that any $d$-collapsible complex is $d$-Leray. Obviously, if $K$ is $d$-Leray, there are no missing $k$-faces in $K$ with $k>d$. Thus, $\mathcal N(\mathcal F)$ being $d$-Leray is a stronger condition than $h(\mathcal F)\leq d+1$.
 
As it turns out, many statements that have first been proven for convex sets and then generalized to families with $d$-collapsible nerve
also hold under the weaker assumptions that the nerve is $d$-Leray, see~ Theorems~\ref{thm:transversal-hypergraph} and \ref{thm:colHel}.

This is very useful, when combined with the following variant of Theorem~\ref{thm:multinerveOne}:
\begin{theorem}[\cite{multinerve13}]
Let $X$ be a locally arc-wise connected topological space and $d$ be the smallest integer such that for all $j\geq d$ and all open sets $G$  we have $\widetilde H_j(G;\Q)=0$.

Assume that $\mathcal F$ is a family of open sets in $X$ with slack $s$ such that for every subfamily $\mathcal G\subseteq F$
of at least $t$ sets, $\bigcap \mathcal G$ has at most $r$ connected components,
then $\mathcal N(\mathcal F)$
is $k$-Leray for $k=r\cdot \left(\max(d,s,t)+1\right)-1$.
\end{theorem}

A first example of a theorem that was originally proven for convex sets and then extended to $d$-Leray complexes is the fractional Helly theorem, originally proved for convex sets by Katchalski and Liu~\cite{fractional79}. The bounds were later improved by Kalai~\cite{Kalai84} to the optimal ones and finally the theorem was proven for all families whose nerve is $d$-Leray.
\begin{theorem}[Fractional Helly theorem for $d$-Leray complexes {\cite[Theorem 12]{transversal-hypergraph}}]\label{thm:transversal-hypergraph}
For every positive integer $d$, the following holds.
If a $d$-Leray simplicial complex $K$ with $v$ vertices
contains at least $\alpha \cdot \binom{v}{d+1}$ faces of dimension $d$, $\alpha\in(0,1)$,
then $K$ contains a face with $\beta_d(\alpha)\cdot v$ vertices, where $\beta_d(\alpha)=1-(1-\alpha)^{\frac{1}{d+1}}$.
\end{theorem}
Hyperplanes in general positions show that, even in convex setting, the number $d+1$ cannot be lowered.

Before we state the next theorem in full generality, we recall a preliminary version for convex sets, which is somewhat easier to understand.
\begin{theorem}[Colorful Helly theorem \cite{lovasz74}]\label{thm:colHelly}
Let $\mathcal C_1$,\ldots, $\mathcal C_{d+1}$ be $d+1$ finite families of convex sets in $\R^d$. If for every choice $C_1\in\mathcal C_1$, $C_2\in\mathcal C_2$, \ldots, $C_{d+1}\in\mathcal C_{d+1}$, $C_1\cap C_2\cap \cdots\cap C_{d+1}\neq \emptyset$, then for some 
$i$, $\bigcap \mathcal C_i\neq \emptyset$. 
\end{theorem}

\begin{theorem}[Colorful Helly theorem for $d$-Leray complexes \cite{lerayCHelly}]\label{thm:colHel}
Let $K$ be a $d$-Leray 
simplicial complex with vertex set $V$. Let $M$ be a matroidal complex on $V$ such that $M\subseteq K$. Let $\rho$ be the rank function of $M$. Then there exists some simplex $\sigma\in K$ for which $\rho(V\setminus \sigma)\leq d$.
\end{theorem}
Although not apparent at first sight, this theorem generalizes Theorem~\ref{thm:colHelly}.
Indeed, let $K$ be the nerve of $\bigcup \mathcal C_i$. By~\cite{wegner75} this complex is $d$-collapsible and thus $d$-Leray.
We define $M$ by declaring that $\rho(S)$ is the smallest $|I|$ such that $S\subseteq \bigcup_{i\in I} \mathcal C_i$. (In other words, $\rho(S)$ counts the number of color classes $\mathcal C_i$ used in $S$.)
The condition $M\subseteq K$ corresponds to the condition that for every choice $C_1\in\mathcal C_1$, $C_2\in\mathcal C_2$, \ldots, $C_{d+1}\in\mathcal C_{d+1}$, $C_1\cap C_2\cap \cdots\cap C_{d+1}\neq \emptyset$, and $\rho(V\setminus \sigma)\leq d$ says that $\sigma$ contains one whole color class $\mathcal C_i$ (hence $\bigcap \mathcal C_i \neq \emptyset$ as $\sigma \in K$).

Theorem \ref{thm:colHel} has found applications in combinatorial geometry and graph theory; see, e.g., \cite{cliques_in_hypergraphs, HolmKar_colorStrongConv, AharoniHolzmanJiang, KimLew}. Kalai and Meshulam further refined Theorem \ref{thm:colHel} to relative setting \cite{KalMesh}.
\bigskip

The following theorem is an exception to the rule that most of theorems that were originally proven for the nerves of convex sets in $\R^d$ have extensions to $d$-Leray complexes. 
So far, the theorem has only been proven for the $d$-collapsible ones.
\begin{theorem}[The optimal colorful fractional Helly theorem for $d$-collapsible complexes \cite{optimalCFHelly}]
Let $K$ be a finite $d$-collapsible simplicial complex and assume that each vertex of $K$ is colored exactly by one of $d+1$ colors $\{1,2,\ldots, d+1\}$. 
Let $N_i$ denote the set of vertices with color $i$,
and let $n_i$ be the cardinality of $N_i$.
Call a face $\sigma$ colorful if all its vertices have distinct colors.

If for some $\alpha\in(0,1]$, $K$ contains at least $\alpha n_1n_2\cdots n_{d+1}$ colorful $d$-faces, then for some $i=1,\ldots, d+1$, $K$ contains a face of dimension at least $\left(1-(1-\alpha)^{1/(d+1)}\right)n_i$, whose all vertices have color $i$.
\end{theorem}

\section{Helly type theorems from non-embeddability}\label{sec:embed}
Another way to prove Helly type theorems is to relate them to non-embeddability results. This approach dates back to Radon, who used Radon's lemma in order to prove Helly's theorem~\cite{Radon1921}.

We formulate the lemma in a way that relates it to the other theorems in this section. 

\begin{lemma}[Radon's lemma~\cite{Radon1921}]\label{lem:radon}
 For every linear map $f$ of $\Delta_{d+1}$ into $\R^d$, there are two disjoint faces $\tau,\sigma$ of $\Delta_{d+1}$ 
 such that $f(\tau)\cap f(\sigma)\neq \emptyset$.
\end{lemma}
Note that if $f$ maps the vertices of $\sigma$
to $v_1,\ldots, v_k$, then by linearity of $f$, $f(\sigma)$ equals $\conv\{v_1,\ldots, v_k\}$.

Helly's theorem is an immediate consequence of this lemma.
\begin{proof}[Proof of Helly's theorem (Thm.~\ref{thm:helly}) from Radon's lemma]
For contradiction assume that the Helly number of convex sets in $\R^d$ is larger than $d+1$. Therefore, there are $n>d+1$ convex sets $G_1,G_2,\ldots, G_n$ for which $G_1\cap\cdots\cap G_n=\emptyset$ and for all $i=1,\ldots, n$, $\bigcap_{j\neq i}G_j\neq\emptyset$. Therefore, for each $i=1,\ldots, n$, we may choose a point $p_i$ that lies in $\bigcap_{j\neq i}G_j$.
Let $f\colon \Delta_{n-1}\to \R^d$ be the linear map that maps the vertex $e_i$ of $\Delta_{n-1}$ to the point $p_i$, $i=1,\ldots, n$.
Then for each non-empty face $\sigma$ of $\Delta_{n-1}$
$f(\sigma)\subseteq \bigcap_{e_j\notin \sigma}G_j$.
Since $n\geq d+2$, Radon's lemma yields two disjoint faces $\sigma,\tau$ of $\Delta_{n-1}$ for which
$f(\sigma)\cap f(\tau)\neq\emptyset$.
But then, the intersection is contained in $G_1\cap \cdots\cap G_n$, which is thus non-empty, a contradiction.
\end{proof}

By using topological Radon's lemma, the assumptions in the Helly theorem can be weakened.
\begin{lemma}[Topological Radon's lemma~\cite{Baj79}]
For every continuous map $f$ of the $(d+1)$-simplex into $\R^d$, there are two disjoint faces $\tau,\sigma$ whose images intersect. 
\end{lemma}
This immediately yields the following Helly-type theorem:
\begin{theorem}\cite{m-httucs-97}\label{thm:hellyMat}
 Let $\mathcal F$ be a finite family of sets in $\R^d$ such that $\bigcap \mathcal F'$ is $d$-connected for every non-empty $\mathcal F'\subseteq\mathcal F$. Then $\bigcap \mathcal F=\emptyset$ if and only if there are at most $d+1$ sets $F_1,F_2,\ldots, F_{d+1}\in\mathcal F$ such that $F_1\cap F_2\cap\cdots\cap F_{d+1}=\emptyset$.
\end{theorem}
\begin{proof}
First we choose the sets $G_1,\ldots, G_n$ as in the previous proof, and then we reach a contradiction.
The only difference is that the map $f$ is not determined by the points $p_i$, but has to be built gradually. First, for each $i\neq j$, choose a path that connects $p_i$ to $p_j$ inside $\bigcap_{k\neq i,j}G_k$. (Such a path exists due to the $d$-connectedness.) Then, when all the paths connecting $p_i$ to $p_j$, $i\neq j$ are chosen,
one needs to choose how to map each triangle such that its boundary agrees with already prescribed part of the map. (Again, this is possible due to the $d$-connectivity.) One continues in this filling, until one reaches a map from the $(d+1)$-skeleton of the $n$-simplex into $\R^d$. The rest is the same as before, finished by the use of the topological Radon's theorem.
\end{proof}
This theorem clearly generalizes Helly's original result. However, if we restrict our attention to open sets only, it is subsumed by Theorems~\ref{thm:montejano} and \ref{thm:multinerveTwo}. 
The advantage of non-embeddability proofs is their greater variability and possibilities of interesting trade-offs. We now illustrate the phenomenon, where instead of the topological Radon's theorem, we use its variant called the van Kampen-Flores theorem and weaken the topological assumptions for the price of a slightly worse bound on $h(\F)$.
\begin{theorem}[\cite{vanKampen:KomplexeInEuklidischenRaeumen-1932}]
 For every continuous map $f$ from the $d$-skeleton of the $(2d+2)$-dimensional simplex into $\R^d$, there are two disjoint faces $\sigma,\tau$ whose images intersect. 
\end{theorem}
\begin{theorem}[\cite{hb17,m-httucs-97}]\label{thm:hVanKampenBound}
 Let $\mathcal F$ be a finite family of sets in $\R^d$ such that $\bigcap \mathcal F'$ is $(\lceil d/2\rceil -1)$-connected for every nonempty $\mathcal F'\subseteq\mathcal F$. Then $\bigcap \mathcal F=\emptyset$ if and only if there are at most $d+2$ sets $F_1,F_2,\ldots, F_{d+2}\in\mathcal F$ such that $F_1\cap F_2\cap\cdots\cap F_{d+2}=\emptyset$.
\end{theorem}
One of the main drawbacks of Theorems~\ref{thm:hellyMat} and \ref{thm:hVanKampenBound} is that by the Adian--Rabin theorem~\cite{Rabin58} and \cite[Prop~1.26]{Hatcher:AlgebraicTopology-2002} it is algorithmically undecidable whether a given space is topologically $d$-connected. It is even undecidable whether $\pi_1(X)$ is trivial.
Using a homological version of Radon's lemma (Lemma~\ref{lem:homRad}) allows us to replace the homotopic assumption of $d$-connectedness with homological conditions, which are algorithmically verifiable\footnote{Provided the input is machine-readable.}.

In order to formulate these results, we introduce several definitions.
A \emph{chain map} $\varphi\colon \widetilde{C}_\bullet(\Delta_n;\Z_2)\to \widetilde{C}_\bullet(\R^d;\Z_2)$ between the augmented chain groups is a collection of linear maps $\varphi_k\colon \widetilde{C}_k(\Delta_n;\Z_2)\to \widetilde{C}_k(\R^d;\Z_2)$, $k=-1,0,1,\ldots$
that respect the boundary operator,
i.e. satisfy $\partial \varphi_k = \varphi_{k-1}\partial$ for all $k\geq 0$.
Here we note that the chain group on the left is the simplicial chain group of the complex $\Delta_n$, whereas on the right we work with the singular chain group of $\R^d$.

\begin{definition}\label{d:non-trivial_chain_map}
The chain map is \emph{non-trivial} if $\varphi_{-1}$ is an isomorphism, i.e., if $\varphi_0$ maps every point to a chain of points with odd number of non-zero coefficients.

If $z=\sum_{i\in I}a_i\sigma_i$ 
and $z'=\sum_{i\in I}b_i\sigma_i$ are two chains in $\widetilde{C}_k(\R^d;\Z_2)$
we say that they \emph{overlap},
if for some $i$ and $j$,
$a_i\neq 0$, $b_j\neq 0$
and\footnote{Note that $\sigma_i$ is a singular $k$-chain, i.e. a continuous map from $\Delta_k$ to $\R^d$. Thus, it makes sense to speak about its image.} $\operatorname{Im}\sigma_i \cap \operatorname{Im}\sigma_j \neq \emptyset$.
If no such pair of indices exists, we say that $z$ and $z'$ are \emph{non-overlapping}. 
\end{definition}

\begin{lemma}[Homological Radon's lemma~\cite{hb17}]\label{lem:homRad}
 For every non-trivial chain map $\varphi$ from the simplicial chains $\widetilde{C}_\bullet(\Delta_{d+1};\Z_2)$ to singular chains $\widetilde{C}_\bullet(\R^d;\Z_2)$ there are two disjoint faces $\sigma$, $\tau$ of $\Delta_{d+1}$ such that $\varphi(\sigma)$ and $\varphi(\tau)$ overlap.
\end{lemma}
This translates to the following Helly-type result.
\begin{theorem}{\cite{hb17}}\label{thm:topHelSharp}
 Let $\mathcal F$ be a finite family of sets in $\R^d$ such that \[\widetilde{H}_i\left(\bigcap \mathcal F';\Z_2\right)=0\text{ for every non-empty }\mathcal F'\subseteq\mathcal F\] and every $i=0,1,\ldots, d-1$. Then $\bigcap \mathcal F=\emptyset$ if and only if there are $d+1$ (not necessarily distinct) sets $F_1,F_2,\ldots, F_{d+1}\in\mathcal F$ such that $F_1\cap F_2\cap\cdots\cap F_{d+1}=\emptyset$.
\end{theorem}
\begin{proof}
 The proof follows the same lines as before. The only difference is that the intersections are no longer required to be topologically connected, but the assumptions about homology allow us to inductively construct a non-trivial chain map, as required.
\end{proof}

Similarly, there is a homological version of the van-Kampen--Flores theorem. 
\begin{theorem}[\cite{hb17}]\label{t:hom_vank_Kampes}
 For every non-trivial chain map $\varphi$ from the simplicial chains $\widetilde{C}_\bullet(\Delta_{2d+2}^{(d)};\Z_2)$ to singular chains $\widetilde C_\bullet(\R^d;\Z_2)$, there are two disjoint simplices $\sigma,\tau$ of $\Delta_{2d+2}^{(d)}$ such that $\varphi(\sigma)$ and $\varphi(\tau)$ overlap.
\end{theorem}
This in turn yields the following Helly-type result.
\begin{theorem}[\cite{hb17}]
 Let $\mathcal F$ be a finite family of sets in $\R^d$ such that $\widetilde{H}_i\left(\bigcap \mathcal F';\Z_2\right)=0$ for all $i=0,1,\ldots, \lceil d/2\rceil-1$ and all nonempty $\mathcal F'\subseteq \mathcal F$. Then $\bigcap \mathcal F=\emptyset$ if and only if there are at most $d+2$ sets $F_1,F_2,\ldots, F_{d+2}\in\mathcal F$ such that $F_1\cap F_2\cap\cdots\cap F_{d+2}=\emptyset$.
\end{theorem}

In contrast to the approaches based on the nerve lemma, the non-embeddability proofs can also deal with the case where the intersections are not connected.
One of the first results in this direction was given by Matoušek~\cite{m-httucs-97}.
\begin{theorem}\cite{m-httucs-97}\label{thm:matousek}
For all non-negative integers $b$ and $d$,
there is a function $f(b,d)$
such that the following holds.
If $\mathcal F$ is a finite family of sets in $\R^d$ such that for all non-empty subfamilies $\mathcal G\subseteq \mathcal F$, $\bigcap \mathcal G$ has at most $b$ path-connected components, each of which is (topologically) $(\lceil d/2\rceil-1)$-connected, then $h(\mathcal F)\leq f(b,d)$.
\end{theorem}
Matoušek's original bound on $f(b,d)$ is huge, as his proof bounds $f(b,d)$ by the hypergraph Ramsey number $R_{b+1}\left(2d+3+\binom{2d+3}{2}(b-1);\binom{b+1}{2}\right)$, that is, by the number $N$ such that in every coloring of $(b+1)$-tuples of an $N$ element set by $\binom{b+1}{2}$ colors, there is always a monochromatic set of size $2d+3+\binom{2d+3}{2}(b-1)$. It is well known that this number $N$ grows approximately as a tower function of height $b$, see~\cite{ramseyBook}.

However, as shown recently~\cite{radonPolynomial}, Matoušek's proof technique can be simplified and the bound improved, even in the homological setting.
\begin{theorem}\cite{radonPolynomial}
For all non-negative integers $b$ and $d$, the following holds.
If $\mathcal F$ is a finite family of sets in $\R^d$, such that for all non-empty subfamilies $\mathcal G\subseteq \mathcal F$, $\bigcap \mathcal G$ has at most $b$ path-connected components and $\widetilde{H}_i(\bigcap\mathcal G;\Z_2)=0$ for all $i=1,2,\ldots, \lceil d/2\rceil -1$, then \[h(\mathcal F)\leq b\cdot \sum_{j=0}^{d+1}\binom{b+1}{2}^j= O(b^{2d+3}).\]
\end{theorem}

In~\cite{hb17} the proof technique and the result of Theorem~\ref{thm:matousek} have been generalized to the following form.
\begin{theorem}[\cite{hb17}]\label{thm:hellyBetti}
 For all positive integers $b$ and $d$, there is a positive integer $h(b,d)$ with the following property.
 Let $\mathcal F$ be a finite family of sets in $\R^d$ such that $\widetilde{\beta}_i\left(\bigcap \mathcal F'; \mathbb Z_2\right)\leq b$ for all non-empty subfamilies $\mathcal F' \subseteq \mathcal F$ and for all $i=0,1,\ldots, \lceil d/2\rceil-1$. Then $h(\mathcal F)\leq h(b,d)$.
\end{theorem}
The proof introduced several novel ideas: homological non-embeddability, constrained chain maps, clever use of barycentric subdivisions, and iterated use of Ramsey's theorem. 
Another impetus came when Patáková realized that the proof of Theorem~\ref{thm:hellyBetti} actually combines two arguments~\cite{bettiRadon}: A proof of a Radon-type result and the classical argument of Levi that a bound on the Radon number implies a bound on the Helly number~\cite{Levi1951}.
As it turns out, a bound on the Radon number is much more universal, since unlike a bound on the Helly number, it implies a bound on fractional Helly numbers \cite{boundedRadon_fractHelly}, Tverberg numbers \cite{jamison1981, Bukh, domotor}, colorful Helly numbers \cite{boundedRadon_fractHelly}, weak $\varepsilon$-nets and $(p,q)$-theorem \cite{transversal-hypergraph, boundedRadon_fractHelly}.
At present, it is not known whether the techniques of \cite{radonPolynomial} and \cite{hb17,bettiRadon} can be combined together.

\subsection{Generalized convexity}
Before we present the technique of constrained chain maps, we introduce a useful notion of generalized convexity.

Given a family $\mathcal C$ of subsets of a topological space $X$ and a set $S\subseteq X$, we define
\[\operatorname{conv}_{\mathcal C}(S):=\bigcap_{S\subseteq F\in\mathcal C} F.\]
If there is no set $F\in\mathcal C$ for which $S\subseteq F$, we set $\operatorname{conv}_{\mathcal C}(S):=X$.

It is not hard to see that $\operatorname{conv}_{\mathcal C}$ is a closure operator (see~\cite{vanDeVel}),
i.e., a map from subsets of $X$ to subsets of $X$ that is 
\begin{description}
\item[Extensive] $S\subseteq \operatorname{conv}_{\mathcal C}S$ for all $S\subseteq X$,
\item[Monotone] $S\subseteq T\subseteq X$ implies $\operatorname{conv}_{\mathcal C}S\subseteq \operatorname{conv}_{\mathcal C}T$ and
\item[Idempotent] $\operatorname{conv}_{\mathcal C}(\operatorname{conv}_{\mathcal C}S)=\operatorname{conv}_{\mathcal C}S$ for all $S\subseteq X$.
\end{description}
The sets that satisfy $\operatorname{conv}_{\mathcal C}S=S$ are called $\mathcal C$-closed.
For example, if $\mathcal C$ is the set of all convex sets, then the $\mathcal C$-closed sets are exactly the usual convex sets.

\begin{theorem}\cite{bettiRadon}\label{t:bettiRadon}
There is a number $r(b,d)$
 with the following property. Let $\mathcal F$ be a family of sets in $\R^d$ such that 
 for every finite set $S\subseteq \R^d$,
 \[\widetilde{\beta}_i(\operatorname{conv}_{\mathcal F}S;\Z_2)\leq b\quad \text{for all }i< \lceil d/2\rceil.\]
 Then in every set $P$ with $r(b,d)$ points there are two disjoint subsets $A,B\subseteq P$ with $\operatorname{conv}_{\mathcal F}A\cap \operatorname{conv}_{\mathcal F}B\neq 0$. 
\end{theorem}

\begin{definition}\cite{Levi1951, bettiRadon}
Let $\mathcal F$ be a family of subsets of $X$. The Radon number $r(\F)$ of $\mathcal F$ is the smallest integer $r$ such that in any set $P\subseteq X$ of cardinality $r$ we find two disjoint subsets $P_1, P_2 \subseteq P$ satisfying $\operatorname{conv}_{\mathcal F}(P_1)\cap \operatorname{conv}_{\mathcal F}(P_2)\neq \emptyset$.
If no such $r$ exists, we put $r(\F) = \infty$.
\end{definition}

As we have seen at the beginning of this section, Radon's lemma (Lemma~\ref{lem:radon}) says that the Radon number of any family $\F$ of convex sets in $\R^d$ is at most $d+2$. This bound is tight, as can easily be seen by considering $d+1$ affinely independent points in $\R^d$.
Radon used this lemma to prove Helly's result. Later, Levi~\cite{Levi1951} showed that the proof works in much greater generality: Any family $\F$ with Radon number $k+1$ has Helly number at most $k$.

Since Theorem \ref{t:bettiRadon}
yields that any family $\mathcal F$ that satisfies the given assumptions has $r(\mathcal F)\leq r(b,d)$, 
Levi's result immediately implies that any such family satisfies $h(\mathcal F)< r(b,d)$,
which is exactly the statement of Theorem~\ref{thm:hellyBetti}, for direct deduction of the result see \cite[Remark 4]{bettiRadon}.

\bigskip
Let us now introduce the main techniques that are used in the proof of Theorem~\ref{t:bettiRadon}:
the constrained chain map method and homological minors (Section~\ref{sec:constrained}); and a Ramsey-type result for closure operators (Section~\ref{sec:ramsey}).
Some of the ideas are already present in the proofs of Theorems~\ref{thm:hellyMat}, \ref{thm:hVanKampenBound} and \ref{thm:topHelSharp}.
Both techniques are of independent interest and may find applications even outside the area of Helly-type theorems.

\subsection{Constrained chain maps \& homological minors}\label{sec:constrained}
Let $\mathbf R$ be a topological space, $K$ a finite simplicial complex and $\gamma_\bullet: \widetilde C_\bullet(K, \mathbb Z_2) \to \widetilde C_\bullet(\mathbf R, \mathbb Z_2)$  a chain map from simplicial chains of $K$ to singular chains of $\mathbf R$. Recall that a chain map $\gamma_\bullet$ is \emph{non-trivial} if the image of every vertex in $K$ is a finite set of points in $\mathbf R$ of odd cardinality, see Definition \ref{d:non-trivial_chain_map}.

\begin{definition}\cite{Uli_HomMinor}
The simplicial complex $K$ is said to be {\em homological minor} of~$\mathbf R$, if there exists a non-trivial chain map $\gamma_\bullet \colon \widetilde C_{\bullet}(K, \mathbb Z_2) \to \widetilde C_{\bullet}(\mathbf R, \mathbb Z_2)$ such that disjoint simplices are mapped to non-overlapping chains. If no such chain map exists we say that $K$ is a {\em forbidden homological minor} of $\mathbf R$.
\end{definition}

\begin{definition}\cite{bettiRadon}\label{def:constrained_chain_map}
Let $\mathcal F$ be a family of sets in~$\mathbf R$ and let $P$ be a set of points in $\mathbf R$. We say that a chain map  $\gamma_\bullet: \widetilde C_\bullet(K, \mathbb Z_2) \to \widetilde C_\bullet(\mathbf R, \mathbb Z_2)$  is \emph{constrained by ($\mathcal F, P)$} if there exists a map $\Phi: K \to 2^P$  such that\footnote{As shown in~\cite{radonPolynomial} it suffices to assume that $\sigma\cap \tau=\emptyset$ $\Rightarrow$ $\Phi(\sigma)\cap \Phi(\tau)=\emptyset$ in the point~(\ref{it:first}). }
\begin{enumerate}
    \item $\Phi(\sigma \cap \tau) = \Phi(\sigma) \cap \Phi(\tau)$ for any $\sigma, \tau \in K$, and $\Phi(\emptyset)=\emptyset$,\label{it:first}
    \item for any simplex $\sigma \in K$, the support of $\gamma_\bullet(\sigma)$ is contained in $\operatorname{conv}_{\mathcal F} \Phi(\sigma)$. \label{it:second}
\end{enumerate}
\end{definition}

Constrained chain maps have the following important property~\cite[Lemma 3.5]{bettiRadon}: if a non-trivial chain map $\gamma_\bullet$ is constrained by $(\mathcal F, P)$, then either $K$ is a homological minor of $\mathbf R$ or there exist two disjoint subsets $A,B$ of $P$ with $\operatorname{conv}_{\mathcal F}A \cap \operatorname{conv}_{\mathcal{F}}B \neq \emptyset$.
Thus, to prove $r(\mathcal F)\leq r$, it suffices to construct, for each $P \subset \mathbf R$ of cardinality $r$, a non-trivial chain map $\gamma_\bullet: \widetilde C_\bullet(K,\mathbb Z_2) \to \widetilde C_\bullet(\mathbf R, \Z_2)$, where $K$ is a forbidden homological minor of $\mathbf R$ and $\gamma_\bullet$ is constrained by $(\mathcal F,P)$. The existence of such a chain map for $\F$ satisfying some additional assumptions is the content of the following statement.

\begin{proposition}\cite[Prop. 3.6 and 3.7]{bettiRadon}\label{p:chain_map}
 Let $K$ be a forbidden homological minor of $\mathbf R$ and let $b$ be a non-negative integer. Then there exists a constant $r$, depending on $K$ and $b$, such that the following holds. 
   For any family of sets $\mathcal F$ in $\mathbf R$ with  \[\widetilde{\beta}_i(\operatorname{conv}_{\mathcal F}S;\Z_2)\leq b\quad \text{for all }i< \dim K \quad \text{and any finite set } S, \] 
   and any $P \subset \mathbf R$ of cardinality $r$, there is a non-trivial chain map $\gamma_\bullet: \widetilde C_\bullet(K,\mathbb Z_2) \to \widetilde C_\bullet(\mathbf R, \Z_2)$ which is constrained by $(\mathcal F,P)$.
\end{proposition}

Theorem \ref{t:bettiRadon} then follows from Proposition \ref{p:chain_map} and Theorem \ref{t:hom_vank_Kampes}, which states that $K = \Delta^{(d)}_{2d+2}$ is a forbidden homological minor of $\mathbb R^d$.

\subsection{Ramsey}\label{sec:ramsey}

 The exposition here is condensed so for details and illustrations we refer to \cite[Section 3.3]{bettiRadon}. 
We fix a set $X$ and positive integers $k, c$.
 For every $V\subseteq X$, let $\rho_V \colon \binom{V}{k}\to [c]$ be a coloring of the $k$-element subsets of $V$. If $|V|<k$,  the coloring $\rho_V$ is, by definition, the empty map.

Let $Y \subseteq X$, $m$ be a positive integer, $M: \binom{Y}{m} \to 2^{X\setminus Y}$, and put $\widehat Z := M(Z)$.
A map $M \colon \binom{Y}{m} \to 2^{X\setminus Y}$ is called \emph{$k$-monochromatic}, if each $m$-element subset $Z$ of $Y$ is monochromatic with respect to the coloring $\rho_{Z \cup \widehat Z} : \binom{Z \cup \widehat Z}{k} \to [c]$, meaning all $k$-element subsets of $Z$ get the same color in the coloring $\rho_{Z \cup \widehat Z}$.

A map $M: \binom{Y}{m} \to 2^{X\setminus Y}$ is \emph{strongly injective} if whenever $Z \neq Z'$, we have $M(Z) \cap M(Z') = \emptyset$ .

\begin{proposition}\cite[Prop. 3.9]{bettiRadon}\label{p:ramsey_selection}
 For any positive integers $k$, $m$, $n$, $c$ there is a constant $N_k=N_k(n;m;c)$ such that the following holds. Let $X$ be a set, and for every $V
\subseteq X$ let $\rho_V \colon \binom{V}{k}\to [c]$ be a coloring of the $k$-element subsets of
$V$. If $|X| \geq N_k$, then there always exists an $n$-element subset $Y \subseteq X$ and a strongly injective $k$-monochromatic map $M: \binom{Y}{m} \to 2^{X\setminus Y}$.
\end{proposition}

\subsection{Proof sketch of Proposition \ref{p:chain_map}.} We provide an extended overview and explain how to combine the Ramsey-type result with the constrained chain map method. For details we refer to \cite{bettiRadon}. The desired chain map $\gamma_\bullet\colon \widetilde C_\bullet(K,\Z_2)\to \widetilde C_\bullet(\mathbf{R}, \Z_2)$ is built by induction on $\dim K$.
If $\dim K$ is zero-dimensional, then $K$ is a collection of vertices, and they are easily mapped to points from $P$.  Note that this gives a condition on cardinality of $P$, it must be at least as large as the number of vertices of $K$.
In the induction step, we proceed as follows. 
Let $t$ denote the dimension of $K$. In order to build the desired map $\gamma_\bullet \colon \widetilde C_\bullet(K;\Z_2) \to \widetilde C_\bullet(\mathbf{R};\Z_2)$, we want to find 
 \begin{gather*}
  \text{a constrained chain map } \chi_\bullet \colon \widetilde C_\bullet(K^{(t-1)};\Z_2) \to \widetilde C_\bullet(\mathbf{R};\Z_2) \text{ such that }\\  \text{ for each $t$-dimensional simplex $\sigma$ of $K$,}\\   \chi_{t-1}(\partial \sigma)  \text{ belongs to the trivial homology class in }  \widetilde H_{t-1}(\conv_{\mathcal F}\Phi(\sigma)). %
\tag{$\clubsuit$}\label{eq:chain_map}
\end{gather*}
Having such a map, its extension to $\gamma_\bullet \colon \widetilde C_\bullet(K;\Z_2) \to \widetilde C_\bullet(\mathbf{R};\Z_2)$ is straightforward. Indeed, since $\chi_{t-1}(\partial \sigma)$ has trivial homology in $\conv_{\mathcal F}\Phi(\sigma)$, it is a boundary, i.e. there exists a chain $\gamma_\sigma\in\widetilde C(\conv_{\mathcal F}\Phi(\sigma))$ such that $\partial \gamma_\sigma = \chi_{t-1}(\partial \sigma)$, so we can 
set $\gamma_k=\chi_k$ for all $k\leq t-1$ and extend it to the whole complex $K$ by setting $\gamma_{t}(\sigma) = \gamma_\sigma$.

It remains to describe the construction of \eqref{eq:chain_map}. Let $\sd K$ stand for the barycentric subdivision of $K$. Then the map $\chi_\bullet$ is obtained as the following composition
\[
\widetilde C_\bullet(K^{(t-1)};\Z_2) \  \xrightarrow[{\color{red}subdivision}]{\makebox[2em]{$\alpha$}}
\  \widetilde C_\bullet\pth{\skel{t-1}{(\sd K)}; \Z_2} \ 
\xrightarrow[{\color{red}Ramsey}]{\makebox[3em]{$\beta$}} \ 
\widetilde C_\bullet\pth{L; \Z_2} \xrightarrow[{\color{red} induction}]{\makebox[2em]{$\gamma'$}} \  \widetilde C_\bullet(\mathbf{R};\Z_2),
\]

where $\alpha$ maps each $\ell$-dim simplex $\sigma$ to a sum of the $\ell$-dim simplices in $\sd \sigma$ and $L$ is a $(t-1)$-dimensional skeleton of a simplex on a sufficiently large number of vertices so that we can use the Ramsey theorem (details later). Let $V(L)$ denote the set of vertices of $L$. As $L$ has dimension $t-1$, we know from induction that for $P$ sufficiently large, there is a non-trivial chain map $\gamma'$ constrained by $(\F, P)$. Let $\Psi_0 \colon L \to 2^P$ witness the constraints. The map $\Psi_0$ is then, in a natural way, extended to $\Psi \colon 2^{V(L)} \to 2^P$, for details see \cite{bettiRadon}. 

We require that $L$ has at least $N_{t+1}(n,m,c)$ vertices, where $N_{t+1}(\cdot,\cdot,\cdot)$ is given by Proposition \ref{p:ramsey_selection}, $n$ is the number of vertices of the barycentric subdivision of $K$, $m$ is the number of vertices of a barycentric subdivision of a $t$-dimensional simplex, and $c$ is the maximal number of elements in $\widetilde H_t(\conv_\F S, \Z_2)$ for $S$ finite. By our assumptions, $c \leq 2^b$.

For each $\tau \subseteq V(L)$, we consider the coloring $\rho_\tau$, which to each $\sigma \subseteq \tau$ on $t+1$ vertices assigns a singular homology class of $\gamma'(\partial \sigma)$ inside $\conv_\F(\Psi(\tau))$. That is, the vertices of $L$ play the role of $X$ in the Ramsey setting. Proposition \ref{p:ramsey_selection}
 gives us an $n$-element subset $Y$ of $V(L)$, and a strongly injective $(t+1)$-monochromatic map.
 Let $\beta$ be induced by the inclusion of $V(\sd K)$ onto $Y$.
Based on $\Psi$ we can define $\Phi: K \to 2^P$, and it can be shown that the strong injectivity yields the condition \eqref{it:first} in Definition \ref{def:constrained_chain_map} for $\Phi$. The choice of $\Psi$ then ensures condition \eqref{it:second} from the same definition, for details see \cite{bettiRadon}.
In other words, the chain map $\chi_\bullet$ is constrained by $(\mathcal F, P)$.
The $(t+1)$-monochromaticity gives that if we take a $t$-dimensional simplex $\tau$ in $K$, then for all $t$-dimensional simplices $\rho\in\sd\tau$  the singular homology class of $\gamma'\beta(\partial \rho)$ inside $\widetilde H_{t-1}(\conv_\F \Phi(\tau))$ is the same. Since the number of such $\rho$'s is even and since we work with $\Z_2$ coefficients, their sum is a cycle with trivial homology, as required in $\eqref{eq:chain_map}$. 
Lastly, the constructed chain map $\chi_\bullet$ is non-trivial as $\gamma'$ is non-trivial.

\subsection*{Optimality.} As we iteratively use a Ramsey-type argument in the proof of Theorem \ref{t:bettiRadon}, we cannot hope that the bound on the constant $r(b,d)$ is optimal. However, the result is optimal in the sense that all $\widetilde \beta_i$, $0 \leq i < \lceil \frac{d}{2}\rceil$ need to be bounded, see \cite[Example~5.1]{bettiRadon}.
The only known bounds for the Radon number in Theorem \ref{t:bettiRadon} are for $b=0$, in which case the Radon number is at most $d+3$ and this is sharp \cite[Theorem 5.2]{bettiRadon}.

We note that Theorem \ref{t:bettiRadon} retains some of the flavor of Theorem \ref{thm:montejano}. More specifically, to reach the conclusion in Theorem~\ref{thm:montejano}, we only need to control the $i$th homology of $\bigcap \mathcal G$
for $\mathcal G$ of a certain size that depends on $i$.
Similarly, from the proof of Theorem~\ref{t:bettiRadon} it follows that, in the case of trivial homology groups, the theorem can be strengthened to the following form.
\begin{theorem}
Let $b=0$ and let $\mathcal F$ be a family of sets in $\R^d$. If for each $i=0,1,\ldots, \lceil d/2 \rceil-1$, $\widetilde{\beta}_{i}(\operatorname{conv}_{\F}S;\Z_2)\leq b$ for all sets $S$ of cardinality $i+2$, then $r(\mathcal F)\leq d+2$. 
\end{theorem}
The statement can be also generalized for $b>0$,
however, the dependence of the cardinality of $|S|$ on $i$ becomes more complicated.

\subsection*{Further results.} Theorem \ref{t:bettiRadon} is formulated for the Euclidean space $\mathbb R^d$, however, the constrained chain map method, as we have seen above, works for any topological space with a forbidden homological minor.  We formulate it as a theorem.

\begin{theorem}\label{t:radonGen}
There is a number $r_K(b)$
 with the following property. Let $K$ be a finite simplicial complex and $\mathbf R$ a topological space such that $K$ is a forbidden homological minor of $\mathbf R$. Let $\mathcal F$ be a family of sets in $\mathbf R$ such that 
 for every finite set $S\subseteq \mathbf R$,
 \[\widetilde{\beta}_i(\operatorname{conv}_{\mathcal F}S;\Z_2)\leq b\quad \text{for all }i< \dim K.\]
 Then the Radon number of $\mathcal F$ is at most $r_K(b)$.
\end{theorem}
For the proof see \cite[Prop. 3.6 and 3.7]{bettiRadon} and the text above.

Notice that the theorem above holds for manifolds as long as we know that there is a forbidden homological minor for them. It is an open question whether for each real manifold there is a forbidden homological minor. 
Similarly, the statement holds for (finite) $k$-dimensional simplicial complexes as each of these complexes embeds in $\R^{2k+1}$. 

Another direction of related results is provided by a beautiful result of Holmsen and Lee \cite{boundedRadon_fractHelly} stating that whenever we have a space $X$ with a closure operator $\conv_\mathcal C$ and a Radon number $r$, the fractional Helly number of $\mathcal C$ is bounded above by some large constant $f(r)$ that depends solely on $r$ and not on $\mathcal C$ itself. For that we note that a family $\mathcal C$ satisfies a \emph{fractional Helly property of order $m$} if there exists
an integer $m$ and a function $\beta: (0,1] \to (0,1]$ for which
the following holds:
Every finite subfamily $\mathcal G \subseteq \mathcal C$ that contains at least $\alpha \binom{|\mathcal G|}{m}$ intersecting $m$ -tuples
contains an intersecting subfamily of size at least $\beta(\alpha)|\mathcal G|$. The smallest integer $m$ such that $\mathcal C$ satisfies the fractional Helly property of order $m$ is called
the \emph{fractional Helly number} for $\mathcal C$. For more details on Helly-type problems in convexity spaces, see a recent survey by Holmsen \cite{holmsen-survey}.

By a standard method of Alon, Kalai, Matoušek and Meshulam \cite{transversal-hypergraph}, bounded fractional Helly number guarantees a $(p,q)$-theorem, and existence of $\varepsilon$-nets, see also  \cite{boundedRadon_fractHelly}. As Theorem~\ref{t:bettiRadon} basically says that under the assumptions on constant homology of intersections we have a bounded Radon number; the above-mentioned theorems extend to this setting as well \cite{boundedRadon_fractHelly, bettiRadon}.  
The main drawback is that the fractional Helly number is very large. Patáková \cite{bettiRadon} reduced it in $\R^2$ to a linear function in $b$ with an optimal value for $b=0$. Further significant improvement was obtained by Goaoc, Holmsen, Patáková \cite{goaoc2024intersectionpatternsspacesforbidden} who focused on the fractional Helly and $(p,q)$-theorems. In fact, they showed that the fractional Helly number is $d+1$ and that the $(p,q)$-theorem holds for $p\geq q\geq d+1$ under the assumptions of bounded homology of intersections.
Their proof is very involved. They introduce a colorful analogue of the "homology" condition and prove a weak colorful Helly theorem. In order to deal with the $m$-partite structure of the colorful classes, they work with cubical complexes instead of simplicial ones, define chain maps inspired by stair-case convexity, and the Ramsey-type statement is replaced by a subgrid lemma \cite[Prop. 4.1]{goaoc2024intersectionpatternsspacesforbidden}, which can be of independent interest.

Furthermore, the following result indicates why $\alpha$ in  fractional Helly-type theorems is assumed to be a fixed constant.
More precisely, if we have a family $\F$ of $n$ sets in the plane such that for each subfamily $\mathcal G\subseteq\F$ its intersection is either empty or path-connected, then by the previous discussion, the fractional Helly number of $\F$ is at most $3$, see \cite{goaoc2024intersectionpatternsspacesforbidden}.
This means that if cubically many triples of sets from $\F$ intersect, a non-trivial fraction of all sets intersect.
The following theorem shows, in a strong form, that assuming only quadratically many intersections of triples is insufficient.
\begin{theorem}[\cite{KalPat}]
For any $n\geq 6$, there is a family $\F$ of $n$ sets in $\R^2$ such that the intersection of every subfamily is either empty or path-connected and for which the following holds:
\begin{itemize}
\item For any $\mathcal G\subseteq \F$ with $|\mathcal G|\geq 4$, $\bigcap \mathcal G=\emptyset$.
\item There are at least $\binom{n-1}{2}+\frac{1}{3}\left( n^2-6n+3\right)$ families $\mathcal G \subseteq \F$
with $|\mathcal G|=3$ and $\bigcap\mathcal G\neq \emptyset$.
\end{itemize}
\end{theorem}

\section{Open problems}\label{sec:open}
We conclude this survey by listing some questions and problems that are closely related to the presented theorems. Some of these questions and many more are mentioned in \cite{Barany_Kalai, bettiRadon, goaoc2024intersectionpatternsspacesforbidden, kalai_conjectures, gil_imre_birthday}. 
\begin{enumerate}
\item In Theorem~\ref{t:bettiRadon}, the upper bound on $r(\mathcal F)$ in terms of $b$ and $d$ is very large. On the other hand, the best known lower bound is $(b+1)(d+2)+1$ as follows from \cite[Example 2]{hb17}. Improve upper and/or lower bounds on the Radon number.
\item Find good lower and upper bounds for the constant $N_k$ in Proposition~\ref{p:ramsey_selection}.
\item The applications of Theorem~\ref{t:radonGen} are hindered by the fact that forbidden homological minors are only known for $\R^d$. 
Find forbidden homological minors for other topological spaces. 

\item Suppose that $\F$ is a family of sets in $\R^d$ such that for any non-empty finite subfamily $\mathcal G$ and every $i=0,1,\ldots, \lceil d/2\rceil -1$, $\widetilde{\beta}_i(\bigcap \mathcal G;R)\leq p(|\mathcal G|)$, where $p$ is some function. 
Does $\F$ enjoy the fractional Helly property of some order which depends only on $p$ and~$d$?
Conjecture 17 in~\cite{gil_imre_birthday} predicts that this is true for $p(|\mathcal G|)$ of the form $\gamma\cdot|\mathcal G|^{d+1}$,
where $\gamma>0$ is a constant. 
The first step was performed in \cite{goaoc2024intersectionpatternsspacesforbidden}, where the conjecture was confirmed for $p$ being a constant function. Subsequently, it was refined for $p$ growing extremely slowly, see ~\cite{bin2024fractionalhellytheoremset}.
\end{enumerate}

\bibliographystyle{alpha}
\bibliography{survey}

\end{document}